\documentclass[copyright,creativecommons]{eptcs}
\providecommand{\event}{MBT 2012} % Name of the event you are submitting to

\usepackage[T1]{fontenc}
\usepackage{textcomp}

\usepackage{enumitem} % For enumeration resuming
\usepackage{amsmath, amssymb, amsthm}  % For various math stuff
\usepackage{ctable}   % For \toprule, \midrule, \bottomrule
\usepackage{tabulary} % For multiline entries in table cells
\usepackage{subfigure} % For floating subfigures
\usepackage{listings} % For source code
\usepackage{centernot} % For the centernot macro
\usepackage{tikz}     % For graphics

\DeclareSymbolFont{letters}     {OML}{ptm}{m}{it}
\DeclareSymbolFont{symbols}     {OMS}{ptm}{m}{n}

\usetikzlibrary{arrows, automata, shapes.geometric}

\newcommand{\abs}[1]{\lvert #1 \rvert}
\newcommand{\aut}[1]{\ensuremath{\mathcal{#1}}}
\newcommand{\set}[1]{\{ \, #1 \, \}}
\newcommand{\tuple}[1]{\langle \, #1 \, \rangle}
\newcommand{\suchthat}{\; . \;}
\newcommand{\project}{\mathrel{\upharpoonright}}
\newcommand{\where}{\mid}
\newcommand{\isomorphic}{\cong}
\newcommand{\union}{\mathrel{\cup}}
\newcommand{\intersection}{\mathrel{\cap}}

\let\oldland\land
\let\oldlor\lor
\let\oldexists\exists
\let\oldforall\forall
\let\oldnexists\nexists

\renewcommand{\land}{\; \oldland \;}
\renewcommand{\lor}{\; \oldlor \;}
\renewcommand{\exists}{\, \oldexists \,}
\renewcommand{\forall}{\, \oldforall \,}
\renewcommand{\nexists}{\, \oldnexists \,}
\newcounter{ruledefcounter}
\newenvironment{ruledef}[2]%
{\vspace{3pt}\noindent\ignorespaces%
\textbf{Rule R\arabic{ruledefcounter}} %
(#1)%
\textbf{:} %
#2\textbf{.}\begin{itemize}\item[]}
{\end{itemize}\stepcounter{ruledefcounter}\vspace{3pt}\ignorespacesafterend}
\theoremstyle{definition}
\newtheorem{definition}{Definition}[section]
\theoremstyle{plain}
\newtheorem{lemma}[definition]{Lemma}
\newtheorem{proposition}[definition]{Proposition}
\newtheorem{theorem}[definition]{Theorem}

\theoremstyle{remark}
\newtheorem{example}[definition]{Example}
\newtheorem{remark}[definition]{Remark}

\newcommand{\theoremlikeParameterised}[2]{\par\medskip\penalty-250\refstepcounter{definition}{\bfseries\scshape\noindent#1 #2.}\slshape}

\newcommand{\linefill}{%
   \cleaders
   \hbox{$\smash{\mkern-2mu\mathord-\mkern-2mu}$}%
   \hfill
   \vphantom{\lower1pt\hbox{$\rightarrow$}}%
}

\newcommand{\Linefill}{%
   \cleaders
   \hbox{$\smash{\mkern-2mu\mathord=\mkern-2mu}$}%
   \hfill
   \vphantom{\hbox{$\Rightarrow$}}%
}

\newcommand{\xleftnoend}[1][]{\mathrel-_{\vphantom{#1}}\mkern-11mu}
\newcommand{\xLeftnoend}[1][]{\mathrel=_{\vphantom{#1}}\mkern-8mu}
\newcommand{\xmid}[2][]{\stackrel{#2}{\linefill_{\vphantom{#1}}}}
\newcommand{\xMid}[2][]{\stackrel{#2}{\Linefill_{\vphantom{#1}}}}
\newcommand{\xrightArrow}[1][]{\mkern-11mu\rightarrow_{#1}}
\newcommand{\xRightArrow}[1][]{\mkern-8mu\Rightarrow_{#1}}
\newcommand{\xmake}[1]{\mathrel{\lower1pt\hbox{$#1$}}}

\newcommand{\trans}[2][]{\xmake{\xleftnoend[#1]\xmid[#1]{#2}\xrightArrow[#1]}}
\newcommand{\ntrans}[2][]{\centernot{\trans[#1]{#2}}}
\newcommand{\etrans}{\rightarrow}

\newcommand{\Trans}[2][]{\xmake{\xLeftnoend[#1]\xMid[#1]{#2}\xRightArrow[#1]}}
\newcommand{\nTrans}[2][]{\centernot{\Trans[#1]{#2}}}
\tikzstyle{every picture}=[->,>=latex,auto,node distance=1.3cm,thick,initial text=,initial where=above,scale=0.7,transform shape]
\tikzstyle{every state}=[draw=white,line width=4pt,fill=black,minimum size=5pt]

\tikzstyle{loop right}=[in=-30,out=30,looseness=8]
\tikzstyle{loop left}=[in=150,out=210,looseness=8]
\tikzstyle{loop above}=[in=60,out=120,looseness=8]
\tikzstyle{loop below}=[in=240,out=300,looseness=8]

\tikzstyle{loop above right}=[in=5,out=65,looseness=8]
\tikzstyle{loop above left}=[in=105,out=165,looseness=8]
\tikzstyle{loop slightly above left}=[in=125,out=185,looseness=8]
\tikzstyle{loop slightly above right}=[in=5,out=65,looseness=8]
\tikzstyle{loop below left}=[in=195,out=255,looseness=8]
\tikzstyle{loop below right}=[in=285,out=-15,looseness=8]
\tikzstyle{loop slightly below left}=[in=155,out=215,looseness=8]
\tikzstyle{loop slightly below right}=[in=305,out=5,looseness=8]
\newcommand{\AparB}{\ensuremath{\aut{A} \parallel \aut{B}}}

\newcommand{\Lin}{L^{\text{\rm I}}}
\newcommand{\Lout}{L^{\text{\rm O}}}
\newcommand{\Ltau}{L^{\tau}}
\newcommand{\Ldeltatau}{L^{\delta}_{\tau}}
\newcommand{\LA}{L_{\aut{A}}}
\newcommand{\LB}{L_{\aut{B}}}

\newcommand{\LinA}{\Lin_{\aut{A}}}
\newcommand{\LinB}{\Lin_{\aut{B}}}

\newcommand{\LoutA}{\Lout_{\aut{A}}}
\newcommand{\LoutB}{\Lout_{\aut{B}}}

\newcommand{\LAparB}{L_{\AparB}}
\newcommand{\LinAparB}{\Lin_{\AparB}}
\newcommand{\LoutAparB}{\Lout_{\AparB}}

\newcommand{\SA}{S_{\aut{A}}}
\newcommand{\SB}{S_{\aut{B}}}

\newcommand{\SAparB}{S_{\AparB}}

\newcommand{\Sstart}{S^0}
\newcommand{\SstartA}{\Sstart_{\aut{A}}}
\newcommand{\SstartB}{\Sstart_{\aut{B}}}

\newcommand{\transA}[2][]{\trans[#1]{#2}_{\aut{A}}}
\newcommand{\transB}[2][]{\trans[#1]{#2}_{\aut{B}}}

\newcommand{\ntransA}[2][]{\ntrans[#1]{#2}_{\aut{A}}}

\newcommand{\etransA}{\etrans_{\aut{A}}}
\newcommand{\etransB}{\etrans_{\aut{B}}}

\newcommand{\etransAparB}{\etrans_{\AparB}}
\newcommand{\etransDelta}{\etrans_{\delta}}
\newcommand{\etransDet}{\etrans_{\mathrm{d}}}
\newcommand{\etransHide}{\etrans_{\mathrm{h}}}

\newcommand{\deltaf}[1]{\ensuremath{\delta(#1)}}
\newcommand{\deter}[1]{\ensuremath{\mathit{det}(#1)}}
\newcommand{\hide}[1]{\ensuremath{\mathit{hide}(#1)}}

\newcommand{\first}[1]{\ensuremath{\mathit{first}(#1)}}
\newcommand{\last}[1]{\ensuremath{\mathit{last}(#1)}}
\newcommand{\trace}[1]{\ensuremath{\mathit{trace}(#1)}}
\newcommand{\paths}[1]{\ensuremath{\mathit{paths}(#1)}}

\newcommand{\reach}[1]{\ensuremath{\mathit{reach}(#1)}}
\newcommand{\reachA}[1]{\ensuremath{\mathit{reach}_{\aut{A}}(#1)}}

\newcommand{\traces}[1]{\ensuremath{\mathit{traces}(#1)}}

\newcommand{\outs}[1]{\ensuremath{\mathit{out}(#1)}}

\title{Talking quiescence: a rigorous theory that supports parallel composition, action hiding and determinisation}

\author{Gerjan Stokkink, Mark Timmer, and Mari\"elle Stoelinga
  \institute{Formal Methods and Tools, Faculty of EEMCS\\University of Twente, The Netherlands}
  \email{\{w.g.j.stokkink, timmer, marielle\}@cs.utwente.nl}
}

\def\titlerunning{Talking quiescence: a rigorous theory that supports parallel composition and determinisation}
\def\authorrunning{Gerjan Stokkink, Mark Timmer, and Mari\"elle Stoelinga}

\begin{document}
\maketitle

\begin{abstract}
The notion of quiescence --- the absence of outputs --- is vital in both behavioural modelling and testing theory. Although the need for quiescence was already recognised in the 90s, it has only been treated as a second-class citizen thus far. This paper moves quiescence into the foreground  and introduces the notion of quiescent transition systems (QTSs): an extension of regular input-output transition systems (IOTSs) in which quiescence is represented explicitly, via quiescent transitions. Four carefully crafted rules on the use of quiescent transitions ensure that our QTSs naturally capture quiescent behaviour.

We present the building blocks for a comprehensive theory on QTSs supporting parallel composition, action hiding and determinisation. In particular, we prove that these operations preserve all the aforementioned rules. Additionally, we provide a way to transform existing IOTSs into QTSs, allowing even IOTSs as input that already contain some quiescent transitions. As an important application, we show how our QTS framework simplifies the fundamental model-based testing theory formalised around ioco.
\end{abstract}

\section{Introduction}

Quiescence is a fundamental concept in modelling system behaviour. It explicitly represents the fact that, in certain system states, no output is provided. The absence of outputs is often essential: an ATM, for instance, should deliver the requested amount of money only once, not twice (see Figure~\ref{fig:atm}). This means that the ATM's state just after paying out money ($s_0$ in Figure~\ref{fig:atm}) should be quiescent: it should not produce any output until further input is given. On the other hand, the state before paying out ($s_3$~in Figure~\ref{fig:atm}) should clearly not be quiescent. Hence, quiescence can also sometimes be considered as erroneous behaviour.

Thus, the notion of quiescence is essential in testing: if a system under test (SUT) does not provide any output, then the test evaluation algorithm must decide whether to produce a pass verdict (allowing quiescence at this point) or a fail verdict (forbidding quiescence at this point).

\paragraph{Origins.}

The notion of quiescence was first introduced by Vaandrager in \cite{Vaandrager91} to obtain a
natural extension of the notion of a terminal or blocking state: if a system is input-enabled (i.e., always ready to receive inputs), then no states are blocking, since each state has outgoing input transitions. However, quiescence can still be used to denote the fact that a state would be blocking when considering only the output actions. Quiescence is explored further in \cite{Segala93, Segala97}.

Tretmans introduced the notion of \emph{repetitive quiescence} \cite{Tretmans96, Tretmans96b}, which emerged from the need to continue testing, even in a quiescent state: in the ATM example above, we need to test further behaviour that arises from the (quiescent) state after providing money. To accommodate these needs, Tretmans introduced the \emph{suspension automaton} as an auxiliary concept. More recent uses of quiescence include~\cite{Aarts10}, applying it in the context of machine learning. 

\begin{example}
Consider the automaton given in Figure~\ref{fig:atm}. The states $s_0$ and $s_1$ are quiescent, since they do not have any outgoing output transitions. To obtain the suspension automaton corresponding to such a system, Tretmans adds self-loops, labelled with the quiescence label $\delta$, to each quiescent state.\qed
\begin{figure}
\centering
\begin{tikzpicture}[node distance=4cm]
	\tikzstyle{every state}=[draw=black,line width=1pt,fill=white,minimum size=5pt]

	\node[initial, state] (s0) {\footnotesize $s_0$};
	\node[state] (s1) [right of=s0] {\footnotesize $s_1$};
	\node[state] (s2) [right of=s1] {\footnotesize $s_2$};
	\node[state] (s3) [right of=s2] {\footnotesize $s_3$};

	\path (s0) edge node {insertCard?}  (s1) 
	    (s1) edge node {re\smash{q}uestMone\smash{y}?}  (s2)
		(s2) edge node {returnCard!} (s3)
		(s3) edge [bend left=10] node {pay!} (s0);
\end{tikzpicture}
\caption{A very basic ATM.}
\label{fig:atm}
\end{figure}
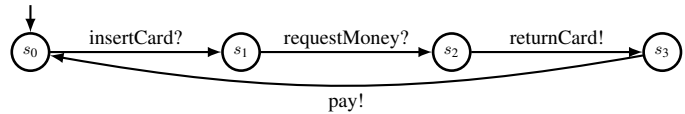
\end{example}

\paragraph{Limitations of current treatments.}

While the papers above all convincingly argued the need for quiescence, none of them presents a comprehensive theory of quiescence. Firstly, quiescence is not treated as a first-class citizen: although the suspension automaton is used during testing, it is not defined as an entity in itself. Therefore, quiescence cannot be used to specify systems, and neither is it clear what properties a suspension automaton satisfies or should satisfy. Since conformance relations such as \texttt{ioco} are defined based on `suspension traces', which are the traces of a suspension automaton, it seems much more appealing to directly start from these suspension automata and base the whole theory on them.

Secondly, basic operators like parallel composition and hiding were only defined for input-output transition systems, but have not been studied for suspension automata at all. Therefore, it was still an open question to what extent these operators could be lifted to the setting of quiescence. 
%We answer this question by carefully crafting four rules.

\paragraph{Our approach.}

The current paper remediates the shortcomings of previous work and presents a comprehensive theory for quiescence, by introducing \emph{quiescent transition systems} (QTSs). These are input-output transition systems in which quiescence can be represented explicitly by $\delta$-transitions, and form a fully-formalised alternative to Tretmans' suspension automata. Whereas suspension automata are always constructed by adding $\delta$-transitions to existing LTSs and subsequently determinising~\cite{Tretmans2008}, QTSs are defined in a precise manner as a stand-alone entity, can be built from scratch and need not necessarily be deterministic. 

As a first step, we handle QTSs that are input-enabled (never reject an input) and most importantly convergent (free of infinite sequences of internal transitions), since the interplay between quiescence and infinite sequences of internal transitions is delicate. Hence, we first focus on the basics. Relaxing these restrictions is an important direction for future work.

Starting point in our theory is the observation that, when treating quiescence as a first-class citizen, restrictions need to be put in place. For instance, it should never be the case that a $\delta$-transition is followed by an output, as this would contradict the meaning of quiescence. As another example, as argued elaborately in Section~\ref{sec:qts}, we do not allow a $\delta$-transition to enable additional behaviour; after all, it would not make much sense if our observation of the absence of outputs impacts the system. In this paper we present and discuss four such rules, that restrict the domain of all possible QTSs to a sensible subclass.

We define three well-known automata-theoretical operations on QTSs: parallel composition, hiding and determinisation. These operations are very important, as they allow a modular approach to system specification. 
Additionally, we explain how to obtain a QTS from an IOTS by a process called deltafication. We define this process in a liberal way, supporting also the construction of a QTS from an IOTS that already has some $\delta$-transitions in place. We show that our four requirements on QTSs, which are a key contribution of this paper, are preserved by all of these operations.

This novel theory of QTSs simplifies the theory of model-based testing. Hence, we conclude this paper by showing how QTSs can be used to define the conformance relation \texttt{ioco}, and aid in test case generation and evaluation.

\paragraph{Overview of the paper.}
First, we present some preliminaries on input-output transition systems in Section~\ref{sec:background}. 
Then, Section~\ref{sec:qts} introduces the QTS model and its operations, as well as a variety of important (closure) properties. Section~\ref{sec:deltafication} explains how to construct QTSs based on IOTSs, 
and Section~\ref{sec:testing} discusses the application of QTSs to test theory. Finally,
conclusions and future work are presented in Section~\ref{sec:conclusions}.

Due to space limitations, we refer to~\cite{techrep} for detailed proofs of all our lemmas, propositions and theorems.

\section{Background}\label{sec:background}

\subsection{Preliminaries}

Given a set $L$, we denote by $L^{*}$ the set of all sequences over $L$.
%, and by $L^{+}$ all nonempty sequences over~$L$. 
Given a sequence $\sigma = a_1a_2 \dotso a_n$, we define the length of $\sigma$, denoted $\abs{\sigma}$, as $n$. The empty sequence is denoted by $\epsilon$.

Given two sequences $\rho = a_1a_2 \dotso a_n \in L^*$ and $\upsilon = b_1b_2 \dotso b_k  \in L^*$, we define the concatenation of $\rho$ and $\upsilon$, denoted $\rho + \upsilon$ or $\rho\upsilon$, as $a_1a_2 \dotso a_nb_1b_2 \dotso b_k$. The sequence $\rho$ is a \emph{prefix} of $\upsilon$, denoted $\rho \sqsubseteq \upsilon$, if there is a $\rho' \in L^*$ such that $\rho\rho' = \upsilon$; if $\rho' \neq \epsilon$, then $\rho$ is a \emph{proper prefix} of~$\upsilon$, denoted $\rho \sqsubset \upsilon$.

Given a set $S \subseteq L^{*}$, a sequence $\sigma \in S$ is called \emph{maximal with respect to $\sqsubseteq$} if there does not exist a sequence $\rho \in S$ such that $\sigma \sqsubset \rho$. Clearly, such a maximal sequence always exists.

We use $\wp(L)$ to denote the \emph{power set} of $L$, i.e., $\wp(L)$ is the set of all subsets of $L$, including the empty set and $L$ itself.

\subsection{Input-Output Transition Systems}

Before we introduce Input-Output Transition Systems, we first describe the modelling formalism they are based on: Labelled Transition Systems.

\begin{definition}[Labelled Transition Systems]
A \emph{Labelled Transition System} (LTS) is a quadruple \linebreak $\aut{A} = \tuple{S, S^0, L, \etrans}$, such that:

\begin{itemize}
\item $S$ is a (possibly uncountable) set of states;
\item $S^0 \subseteq S$ is a non-empty set of initial states;
\item $L$ is a set of labels, each representing a different action. We take $\tau \notin L$ to stand for an internal (unobservable) action and define $\Ltau = L \cup \{\tau\}$;
\item $\etrans {} \subseteq S \times \Ltau \times S$ is the transition relation. We use $s \trans{a} s'$  to denote $(s, a, s') \in{} \etrans {}$, write $s \trans{a}$ if there is an $s' \in S$ such that $s \trans{a} s'$, and $s \ntrans{a}$ if this is not the case. If $s \trans{a}$, we say that the action $a$ is \emph{enabled} in state $s$.
\end{itemize}
\end{definition}

We use $\SA$, $\SstartA$, $\LA$ and $\etransA$ to denote the components of an LTS \aut{A}. These subscripts are left out when it is clear from the context which LTS is referred to.

\begin{example}
Figure~\ref{fig:lts_and_iots_example}(a) shows an LTS \aut{A}. \qed

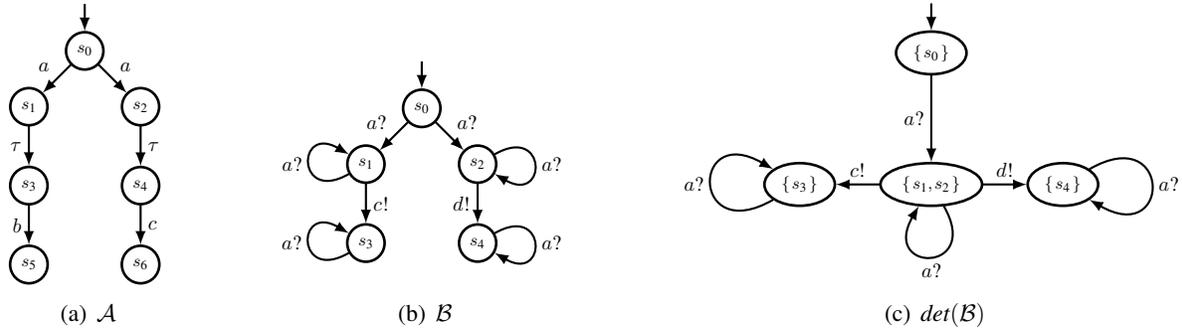
\begin{figure}
\subfigure[\aut{A}]{
\begin{tikzpicture}[node distance=1.5cm]
	\tikzstyle{every state}=[draw=black,line width=1pt,fill=white,minimum size=5pt]

	\node[initial, state] (s0) {\footnotesize $s_0$};
	\node[state] (s1) [below left of=s0] {\footnotesize $s_1$};
	\node[state] (s2) [below right of=s0] {\footnotesize $s_2$};
	\node[state] (s5) [below of=s1] {\footnotesize $s_3$};
	\node[state] (s3) [below of=s5] {\footnotesize $s_5$};
	\node[state] (s4) [below of=s2] {\footnotesize $s_4$};
	\node[state] (s6) [below of=s4] {\footnotesize $s_6$};

	\path (s0) edge node [swap]  {$a$}  (s1)
		(s0) edge node {$a$}  (s2)
		(s1) edge node [swap] {$\tau$} (s5)
		(s5) edge node [swap] {$b$}  (s3)
		(s2) edge node {$\tau$} (s4)
		(s4) edge node {$c$} (s6);
\end{tikzpicture}
}
\hfill
\subfigure[\aut{B}]{
\begin{tikzpicture}[node distance=1.5cm]
	\tikzstyle{every state}=[draw=black,line width=1pt,fill=white,minimum size=5pt]
	
	\node[initial, state] (s0) {\footnotesize $s_0$};
	\node[state] (s1) [below left of=s0] {\footnotesize $s_1$};
	\node[state] (s2) [below right of=s0] {\footnotesize $s_2$};
	\node[state] (s3) [below of=s1] {\footnotesize $s_3$};
	\node[state] (s4) [below of=s2] {\footnotesize $s_4$};

	\path (s0) edge node [swap]  {$a?$}  (s1)
		(s0) edge node {$a?$}  (s2)
		(s1) edge [loop left] node {$a?$} (s1)
		(s1) edge node {$c!$}  (s3)
		(s2) edge [loop right] node {$a?$} (s2)
		(s2) edge node [swap] {$d!$} (s4)
		(s3) edge [loop left] node {$a?$} (s3)
		(s4) edge [loop right] node {$a?$} (s4);
\end{tikzpicture}
}
\hfill
\subfigure[\deter{\aut{B}}]{
\begin{tikzpicture}[node distance=2.5cm]
	\tikzstyle{myellipse} = [ellipse, draw=black,line width=1pt,fill=white,minimum size=5pt]
	
	\node[initial, myellipse] (s0) {\footnotesize $\set{s_0}$};
	\node[myellipse] (s1) [below of=s0] {\footnotesize $\set{ s_1, s_2}$};
	\node[myellipse] (s2) [left of=s1] {\footnotesize $\set{s_3}$};
	\node[myellipse] (s3) [right of=s1] {\footnotesize $\set{s_4}$};

	\path 	(s0) edge node [swap]  {$a?$}  (s1)
				(s1) edge [loop below] node {$a?$} (s1)
				(s1) edge node [swap] {$c!$}  (s2)
				(s1) edge node {$d!$}  (s3)
				(s2) edge [loop left] node {$a?$} (s2)
				(s3) edge [loop right] node {$a?$} (s3);
\end{tikzpicture}
}
\caption{Visual representation of the LTS \aut{A} and the IOTSs \aut{B} and \deter{\aut{B}}. We represent states by circles, and transitions by arrows; each arrow in turn is labelled with the associated action for that particular transition. The initial state is marked by an arrow without a source state. Frow now on, we typically will not label individual states.}
\label{fig:lts_and_iots_example}
\end{figure}
\end{example}

Often, in particular in the context of testing, it is desirable to be able to distinguish between actions that are initiated by the environment (inputs), and actions that are initiated by the system itself (outputs). To this end, we introduce Input-Output Transition Systems, which are an extension of regular LTSs.

\begin{definition}[Input-Output Transition Systems]
An \emph{Input-Output Transition System} (IOTS) is a quintuple $\aut{A} = \tuple{S, S^0, \Lin, \Lout, \etrans}$, where $\Lin$ is a set of input labels and $\Lout$ a set of output labels such that $\Lin \intersection \Lout = \emptyset$. We define $L = \Lin \union \Lout$ and $\Ltau = L \cup \set{ \tau }$, where $\tau \notin L$. $S$, $S^0$ and $\etrans{}$ are as defined for LTSs. Additionally, IOTSs must be \emph{input-enabled}, i.e., $s \trans{a}$ for all $s \in S, a \in \Lin$.\end{definition}

\begin{remark}
Throughout this article we sometimes suffix a question mark ($?$) to the input labels and an exclamation mark ($!$) to the output labels, to help differentiating the two types. These are, however, not part of the label.
\end{remark}

Note that IOTSs are similar to I/O automata~\cite{LT87,DeNicola1995}, except that the latter allow multiple internal actions, rather than $\tau$ only. All our results can easily be phrased in the I/O automata framework.

By requiring IOTSs to be input-enabled, any input initiated by the environment is never refused by the system. For deterministic systems (see Definition~\ref{def:determinism}), this restriction can easily be lifted by adding a sink state which has self-loops for all possible actions, and adding transitions for the missing inputs to that sink state (so-called \emph{demonic completion} \cite{DeNicola1995,Bijl2004}). For nondeterministic systems, a solution is provided in~\cite{BS08}.

\begin{example}
Figure~\ref{fig:lts_and_iots_example}(b) shows an IOTS \aut{B}. Note that since $\Lin = \set{a}$ and $s \trans{a}$ for every $s \in S$, \aut{B} is input-enabled. \qed
\end{example}

We introduce the standard language-theoretic concepts for IOTSs.

\begin{definition}[Notations]
\label{def:notations}
Let $\aut{A} = \tuple{S, S^0, \Lin, \Lout, \etrans}$ be an IOTS, then:

\begin{itemize}
\item A \emph{path} in \aut{A} is a (possibly infinite) sequence $\pi = s_0 a_1 s_1 \dotso s_n$ such that for all $1 \leq i \leq n$ we have $s_{i-1} \trans{a_i} s_i$ with $a_i \in \Ltau$. The set of all paths in $\aut{A}$ is denoted $\paths{\aut{A}}$.

\item The path operators $\mathit{first}$ and $\mathit{last}$ yield the first and last state of a finite path, respectively, e.g., for $\pi = s_0 a_1 s_1 a_2 s_2$ we have $\first{\pi} = s_0$ and $\last{\pi} = s_2$. A path $\pi$ is called \emph{initial} if $\first{\pi} \in S^0$.

\item The path operator $\mathit{trace}$ yields the sequence of actions that is obtained by erasing all states and $\tau$-actions from a given path, e.g., for $\pi = s_0 a_1 s_1 \tau s_2 a_2 s_3$ we have $\trace{\pi} = a_1 a_2$; we call such a sequence of actions a $\mathit{trace}$ of $\aut{A}$. The \emph{length} of a trace $\sigma = a_1 a_2 \dotso a_n$, denoted $\abs{\sigma}$, is the length of the corresponding sequence, i.e., $\abs{\sigma} = \abs{a_1 a_2 \dotso a_n} = n$.

\item Given an action $a$ and a set of actions $P$, we denote by $a \project P$ the \emph{projection} of $a$ on $P$, i.e., $a \project P = a$ if $a \in P$, and $a \project P = \epsilon$ otherwise. The projection of a trace $\sigma = a\sigma'$ on a set of actions $P$ follows naturally from this: $\sigma \project P = a\sigma' \project P = a \project P \; + \; \sigma' \project P$. Finally, the projection of a set of traces $T$ on a set of actions $P$ is defined as $T \project P = \set{ \sigma \project P \where \sigma \in T}$.

\item If there is a finite path $\pi$ in $\aut{A}$ such that $\first{\pi} = s$, $\last{\pi} = s'$ and $\trace{\pi} = \sigma$, we write $s \Trans{\sigma} s'$; if there exists an $s' \in S$ such that $s \Trans{\sigma} s' $, we write $s \Trans{\sigma}$, and $s \nTrans{\sigma}$ if this is not the case. 

\item For a finite trace $\sigma$ and state $s \in S$, we denote by $\reach{s, \sigma}$ the set of states in \aut{A} (possibly empty) that can be reached from $s$ via $\sigma$, i.e., $\reach{s, \sigma} = \set{ s' \in S \where s \Trans{\sigma} s' }$. Similarly, for a given finite trace $\sigma$ and a set of states $S' \subseteq S$, we denote by $\reach{S', \sigma}$ the set of states in \aut{A} that can be reached from any of the states in $S'$ via $\sigma$, i.e., $\reach{S', \sigma} = \set{ s \in S \where \exists s' \in S' \suchthat s' \Trans{\sigma} s }$.

\item For a finite trace $\sigma$ and state $s \in S$, $\outs{s, \sigma}$ is the set of output actions that are enabled in any of the states reachable from $s$ by $\sigma$, i.e., $\outs{s, \sigma} = \set{ a \in \Lout \where \exists s' \in \reach{s, \sigma} \suchthat s' \Trans{a}}$. We use the shorthand $\outs{s}$ for the case $\outs{s, \epsilon}$, i.e., the set of output actions that are enabled in $s$ itself.

\item For every $s \in S$ we denote by $\traces{s}$ the set of all traces of $\aut{A}$ that correspond to paths that start in $s$, i.e., $\traces{s} = \set{ \trace{\pi} \where \pi \in \paths{\aut{A}} \land \first{\pi} = s }$. We denote by $\traces{\aut{A}} = \bigcup_{s \in S^0} \; \traces{s}$ the set of all traces that correspond to initial paths in $\aut{A}$. Two IOTSs \aut{B} and \aut{C} are \emph{trace equivalent}
%, denoted \aut{B} \treq \aut{C}, 
if $\traces{\aut{B}} = \traces{\aut{C}}$.
\end{itemize}
\end{definition}

A fundamental concept in automata theory is determinism.

\begin{definition}[Determinism]
\label{def:determinism}
An IOTS $\aut{A} = \tuple{S, S^0, \Lin, \Lout, \; \trans{}}$ is \emph{deterministic} if for all $s,s',s''\in S,  a \in L$ we have that $s \trans{a} s'$ and  $s \trans{a} s''$ imply $a \neq \tau$ and $s' = s''$. Otherwise, \aut{A} is \emph{nondeterministic}.
\end{definition}

\begin{figure}[t]
\centering
\subfigure[\aut{A}]{
\begin{tikzpicture}
	\node[initial, state] (s0) {};
	\node[state] (s1) [below left of=s0] {};
	\node[state] (s2) [below right of=s0] {};

	\path (s0) edge node [swap]  {$a!$}  (s1)
		  	 (s0) edge node {$b!$}  (s2);
\end{tikzpicture}
}
\subfigure[\aut{B}]{
\begin{tikzpicture}
	\node[initial, state] (s0) {};
	\node[state] (s1) [below left of=s0] {};
	\node[state] (s2) [below right of=s0] {};

	\path (s0) edge node [swap]  {$a!$}  (s1)
		  	 (s0) edge node {$\tau$}  (s2);
\end{tikzpicture}
}
\subfigure[\aut{C}]{
\begin{tikzpicture}
	\node[initial, state] (s0) {};
	\node[state] (s1) [below left of=s0] {};
	\node[state] (s2) [below right of=s0] {};

	\path 	(s0) edge node [swap]  {$a!$}  (s1)
		  		(s0) edge node {$a!$}  (s2);
\end{tikzpicture}
}
\subfigure[\aut{D}]{
\begin{tikzpicture}
	\node[initial, state] (s0) {};
	\node[state] (s1) [below right of=s0] {};
	\node[state] (s2) [below left of=s1] {};
	\node[state] (s3) [below left of=s0] {};
	\node[state] (s4) [below right of=s1] {};
	\node[state] (s5) [below left of=s3] {};
	
	\path		(s0) edge node {$\tau$} (s1)
				(s1) edge node {$\tau$} (s2)
				(s1) edge node {$a!$} (s4)
				(s2) edge node {$\tau$} (s3)
				(s3) edge node {$\tau$} (s0)
				(s3) edge node [swap] {$b!$} (s5);
\end{tikzpicture}
}
\caption{One deterministic (\aut{A}) and three nondeterministic (\aut{B}, \aut{C}, \aut{D}) IOTSs. The IOTS \aut{D} is divergent.}
\label{fig:iots_nondeterministic_and_divergent}
\end{figure}
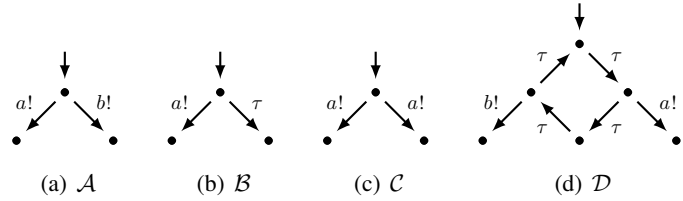

\begin{example}
Figure~\ref{fig:iots_nondeterministic_and_divergent} shows some deterministic and nondeterministic IOTSs. \qed
\end{example}

Lastly, we introduce the notions of convergence and divergence.

\begin{definition}[Divergence]
Given an IOTS $\aut{A} = \tuple{S, S^0, \Lin, \Lout, \; \trans{}}$, a state $s \in S$ of \aut{A} is \emph{divergent} if there is an infinite path $s_0 a_1 s_1 a_2 s_2 \dotso$ with $s_0 = s$ and $s_i \in S$, that contains only $\tau$ transitions, i.e., $a_i=\tau$ for all $i$. An IOTS is called divergent if it contains at least one such state, otherwise it is \emph{convergent}.
\end{definition}

For the purposes of this paper, we require all IOTSs to be convergent.

\begin{example}
Figure~\ref{fig:iots_nondeterministic_and_divergent}(d) shows the divergent IOTS \aut{D}. Clearly, it is possible for \aut{D} to perform an infinite sequence of $\tau$-transitions by continuously looping through the innermost four states. \qed
\end{example}

\subsection{Operations on IOTSs}
In this section, we introduce several standard operations on IOTSs. First, every nondeterministic IOTS can be transformed into a deterministic IOTS \cite{sudkamp2006}; the latter is called the determinisation of the original IOTS and is trace equivalent to it \cite{katoen2008}. Using this operator, modelling effort is saved since no attention needs to be paid to making the specification deterministic.

\begin{definition}[Determinisation]
\label{def:determinisation}
The \emph{determinisation} of an IOTS $\aut{A} = \tuple{S, S^{0}, \Lin, \Lout, \etransA}$ is the IOTS $\deter{\aut{A}} = \tuple{T, \set{S^0}, \Lin, \Lout, \etransDet}$ such that $T = \wp{(S)} \setminus { \emptyset }$ and $\etransDet{} = \set{(U, a, V) \in T \times L \times T \where V = \linebreak \reachA{U, a} \land V \neq \emptyset }$.
\end{definition}

\begin{example}
Consider the nondeterministic IOTS \aut{B} shown in Figure~\ref{fig:lts_and_iots_example}(b). Its corresponding determinisation \deter{\aut{B}} is shown in Figure~\ref{fig:lts_and_iots_example}(c). \qed
\end{example}

Second, we define the parallel composition operator. This operator is fundamental in modelling frameworks for component-based design. It allows one to build complex system models from smaller ones, thus breaking up the specification of a system into manageable pieces. Parallel composed IOTSs synchronise on shared inputs and complementary input-output pairs~\cite{DeNicola1995}.

\begin{definition}[Parallel composition of IOTSs]
\label{def:parallel_iots}
Given are two IOTSs $\aut{A} = \tuple{\SA, S^{0}_{\aut{A}}, \LinA, \LoutA, \etransA}$ and $\aut{B} = \tuple{\SB, S^{0}_{\aut{B}}, \LinB, \LoutB, \etransB}$ such that $\LoutA \intersection \LoutB = \emptyset$. The \emph{parallel composition} of \aut{A} and \aut{B} is the IOTS $\AparB = \tuple{\SAparB, S^0_{\AparB}, \LinAparB, \LoutAparB, \etransAparB}$, where $\SAparB = S_{\aut{A}} \times S_{\aut{B}}$, $S^0_{\AparB} = S^{0}_{\aut{A}} \times S^{0}_{\aut{B}}$, $\LinAparB = \linebreak (\LinA \union \LinB) \setminus (\LoutA \union \LoutB)$, and $\LoutAparB = \LoutA \union \LoutB$. The transition relation $\etransAparB$ is defined as follows:
\begin{align*}
\etransAparB \; \; = \; \; & \set{ ((s, t), a?, (s', t')) \where s \transA{a?} s' \land t \transB{a?} t' } \\
\union \; \; & \set{ ((s, t), a!, (s', t')) \where s \transA{a?} s' \land t \transB{a!} t' } \\
\union \; \; & \set{ ((s, t), a!, (s', t')) \where s \transA{a!} s' \land t \transB{a?} t' } \\
\union \; \; & \set{ ((s, t), a, (s', t)) \where s \transA{a} s' \land t \in \SB \land a \in \Ltau_{\aut{A}} \setminus \LB } \\
\union \; \; & \set{ ((s, t), a, (s, t')) \where t \transB{a} t' \land s \in \SA \land a \in \Ltau_{\aut{B}} \setminus \LA }
\end{align*}
\end{definition}

Thus, $\LAparB =  \LinAparB \union \LoutAparB = \LA \union \LB$.

\begin{figure}[t]
\subfigure[\aut{A}] {
\begin{tikzpicture}
	\node[initial, state] (s0) {};
	\node[state] (s1) [below of=s0] {};
	\node[state] (s2) [below of=s1] {};
	\node[state] (s3) [below of=s2] {};

	\path		(s0) edge [loop above right] node {$b?, c?$} (s0)
				(s0) edge node [swap] {$a?$} (s1)
				(s1) edge [loop above right] node {$a?, c?$} (s1)
				(s1) edge node [swap] {$b?$} (s2)
				(s2) edge [loop left] node {$a?, b?, c?$} (s2)
				(s2) edge node {$d!$} (s3)
				(s3) edge [loop left] node {$a?, b?, c?$} (s3);
\end{tikzpicture}
}
\hfill
\subfigure[\aut{B}] {
\begin{tikzpicture}
	\node[initial, state] (s0) {};
	\node[state] (s1) [below of=s0] {};
	\node[state] (s2) [right of=s1] {};
	\node[state] (s3) [below of=s1] {};
	\node[state] (s4) [below of=s2] {};
	\node[state] (s5) [below of=s3] {};
	\node[state] (s6) [below of=s4] {};

	\path		(s0) edge [loop above right] node {$b?$} (s0)
				(s0) edge [swap] node {$a!$} (s1)
				(s0) edge node {$d?$} (s2)
				(s1) edge [loop left] node {$d?$} (s1)
				(s1) edge [swap] node {$b?$} (s3)
				(s1) edge node {$c!$} (s4)
				(s2) edge [loop right] node {$b?,d?$} (s2)
				(s3) edge [loop left] node {$b?,d?$} (s3)
				(s3) edge node {$e!$} (s5)
				(s5) edge [loop left] node {$b?,d?$} (s5)
				(s4) edge [loop right] node {$b?,d?$} (s4)
				(s4) edge node [swap] {$a!$} (s6)
				(s6) edge [loop right] node {$b?,d?$} (s6);
\end{tikzpicture}
}
\hfill
\subfigure[$\aut{A} \parallel \aut{B}$] {
\begin{tikzpicture}
	\node[initial, state] (s0) {};
	\node[state] (s1) [right of=s0] {};
	\node[state] (s2) [below of=s1] {};
	\node[state] (s3) [right of=s1] {};
	\node[state] (s4) [below of=s2] {};
	\node[state] (s5) [below of=s3] {};
	\node[state] (s6) [left of=s4] {};
	\node[state] (s7) [below of=s6] {};
	\node[state] (s8) [below of=s4] {};
	\node[state] (s9) [below of=s5] {};
	\node[state] (s10) [right of=s3] {};
	\node[state] (s11) [below of=s10] {};
	\node[state] (s12) [below of=s11] {};
	\node[state] (s13) [right of=s11] {};
	\node[state] (s14) [below of=s13] {};
	\node[state] (s15) [below of=s9] {};

	\path		(s0) edge node {$a!$} (s1)
				(s1) edge [swap] node {$b?$} (s2)
				(s1) edge node {$c!$} (s3)
				(s2) edge node {$e!$} (s4)
				(s2) edge node [swap] {$d!$} (s6)
				(s3) edge node [swap] {$a!$} (s5)
				(s3) edge node {$b?$} (s10)
				(s4) edge node {$d!$} (s8)
				(s5) edge node [swap] {$b!$} (s9)
				(s6) edge node [swap] {$e!$} (s7)
				(s9) edge node [swap] {$d!$} (s15)
				(s10) edge node {$d!$} (s13)
				(s10) edge node [swap] {$a!$} (s11)
				(s11) edge node [swap] {$d!$} (s12)
				(s13) edge node {$a!$} (s14);
\end{tikzpicture}
}
\caption{The IOTSs \aut{A} and \aut{B}, and their parallel composition $\AparB$. Note that we have left out some of the $b?$-labelled self-loops from the visualisation of $\AparB$ to reduce clutter.}
\label{fig:iots_parallel_composition_example}
\end{figure}

\begin{example}
Figure~\ref{fig:iots_parallel_composition_example} shows two IOTSs \aut{A} and \aut{B}, and their parallel composition $\AparB$. We have \linebreak$\LinA = \set{a, b, c}$, $\LoutA = \set{d}$, $\LinB = \set{b, d}$, and $\LoutB = \set{a, c, e}$. Note that indeed $\LoutA \intersection \LoutB = \emptyset$, as required; therefore, by Definition \ref{def:parallel_iots}, $\LinAparB = \set{b}$ and $\LoutAparB = \set{a, c, d, e}$. \qed
\end{example}

Finally, it is often useful to hide certain actions of a given IOTS, thereby essentially renaming the corresponding labels to $\tau$. For example, when parallel composing two IOTSs, some actions are only used for synchronisation; after composition, they are not needed anymore.

\begin{definition}[Action hiding in IOTSs]
\label{def:iots_action_hiding}
Let $\aut{A} = \tuple{S, S^0, \Lin, \Lout, \etransA}$ be an IOTS and $H \subseteq \Lout$ a set of output labels, then one can \emph{hide} $H$ in \aut{A} to get the IOTS $\hide{\aut{A}, H} = \tuple{S, S^{0}, \Lin, \Lout \setminus H, \etransHide}$, where $\etransHide {} = \set{ (s, a, s') \in \etransA \where a \notin H} \union \set{ (s, \tau, s') \in S \times \set{\tau} \times S \where \exists a \in H \suchthat (s, a, s') \in \etransA}$.
\end{definition}

Thus, we only allow output actions to be hidden. Furthermore, we do not allow action hiding to lead to divergent IOTSs, i.e., the hiding of outputs may not lead to the creation of $\tau$-loops.

\begin{figure}[t]
\centering
\subfigure[\aut{A}]{
\begin{tikzpicture}
	\node[initial, state] (s0) {};
	\node[state] (s2) [below of=s0] {};
	\node[state] (s1) [left of=s2] {};
	\node[state] (s3) [right of=s2] {};
	\node[state] (s4) [below of=s2] {};

	\path	(s0) edge node [swap] {$a!$} (s1)
			(s0) edge node [swap] {$\tau$} (s2)
			(s0) edge node {$c!$} (s3)
			(s2) edge node [swap] {$b!$}  (s4);
\end{tikzpicture}
}
\subfigure[$\hide{\aut{A}, \set{ a, b }}$]{
\begin{tikzpicture}
	\node[initial, state] (s0) {};
	\node[state] (s2) [below of=s0] {};
	\node[state] (s1) [left of=s2] {};
	\node[state] (s3) [right of=s2] {};
	\node[state] (s4) [below of=s2] {};

	\path	(s0) edge node [swap] {$\tau$} (s1)
			(s0) edge node [swap] {$\tau$} (s2)
			(s0) edge node {$c!$} (s3)
			(s2) edge node [swap] {$\tau$}  (s4);
\end{tikzpicture}
}
\caption{The IOTSs \aut{A} and $\hide{\aut{A}, \set{ a, b }}$.}
\label{fig:hiding_example}
\end{figure}

\begin{example}
Figure~\ref{fig:hiding_example} shows the IOTSs \aut{A} with $\Lout_{\aut{A}} = \set{a, b, c}$ and $\aut{B} = \hide{\aut{A}, \set{ a, b }}$. \qed
\end{example}

From now on, we typically won't show all input-labelled self-loops in visualisations of IOTSs, to reduce clutter. Thus, we assume that every IOTS is input-enabled (unless mentioned otherwise).

\subsection{Properties of IOTSs}\label{sec:iots_properties}
IOTSs possess several interesting properties, that will also be of use when working with QTSs later on.
We provide three results, showing that (1)~hiding of actions corresponds to projection of traces, (2)~parallel composition does not introduce new traces when projecting on the alphabet of either one of the components, and (3)~parallel composition of components that synchronise on all actions yields the intersection of the traces of the components.
\newcommand{\propositionIOTSHidingAndProjection}{%
Given an IOTS \aut{A} and a set of labels $H \subseteq \LoutA$, we have $\traces{\hide{\aut{A}, H}} = \linebreak \traces{\aut{A}} \project (\LA \setminus H)$.
}

\begin{proposition}
\label{prop:hiding_and_projection}
\propositionIOTSHidingAndProjection
\end{proposition}

\newcommand{\propositionIOTSParallelTraceInclusion}{%
Given two IOTSs \aut{A} and \aut{B}, we have $\traces{\AparB} \project \LA \subseteq \traces{\aut{A}}$ and \linebreak $\traces{\AparB} \project \LB \subseteq \traces{\aut{B}}$.
}

\begin{proposition}
\label{prop:iots_parallel_trace_inclusion}
\propositionIOTSParallelTraceInclusion
\end{proposition}

\newcommand{\propositionIOTSParallelTraceEquivalence}{%
Given two IOTSs \aut{A}, \aut{B} with $\LA = \LB$, we have $\traces{\AparB} = \traces{\aut{A}} \intersection \traces{\aut{B}}$.
}

\begin{proposition}
\label{prop:iots_parallel_trace_equivalence}
\propositionIOTSParallelTraceEquivalence
\end{proposition}

\section{Quiescent Transition Systems}\label{sec:qts}

\subsection{Basic notions and requirements}

IOTSs can be used to model the inputs and outputs of a system, but cannot explicitly express the observation of the absence of outputs, also called the observation of quiescence \cite{Vaandrager91, Tretmans96, Segala97}. To fill this void, we introduce Quiescent Transition Systems. 
These automata can be used to model all possible observations for a particular system, including quiescence, and can thus be thought of as `observation automata'. They are based on Tretmans' suspension automata~\cite{Tretmans96}, in the sense that a $\delta$-transition represents the observation of quiescence.
%, and can also  offer a precise semantics for the concept of the refusal transition systems introduced in \cite{Bour2010}. 
A basic variant of QTSs was already used in~\cite{Timmer2011} in a testing framework. However, restrictions for QTSs to prohibit counterintuitive behaviour, as well as characteristics and closure properties of such models, have never been studied before.

\begin{definition}[Quiescence]
Let $\aut{A} = \tuple{S, S^0, \Lin, \Lout, \etrans}$ be an IOTS. A state $s \in S$ is called \emph{quiescent} if $\nexists a \in \Lout \union \set{\tau} \suchthat s \trans{a}$, i.e., no outputs or internal transitions can be executed in state $s$. 
\end{definition}
A system in a quiescent state will be idle until a new input is supplied. Note that a state $s$ that can still perform a $\tau$-step is not considered quiescent, even if there is no output $a! \in \Lout$ such that $s \Trans{a!}$. After all, since quiescence signifies that a system is idle
indefinitely, it would not make sense if there are still internal steps possible. Moreover, from a more technical point of view, this ensures that QTSs are closed under hiding and that 
hiding and deltafication (see Section~\ref{sec:deltafication}) are commutative.

\begin{figure}[t]
\subfigure[\aut{A}]{
\begin{tikzpicture}
	\node[initial, state] (s0) {};
	\node[state] (s1) [below left of=s0] {};
	\node[state] (s2) [below right of=s0] {};

	\path (s0) edge node [swap]  {$a?$}  (s1)
		(s0) edge node {$b?$}  (s2)
		(s1) edge [loop left] node {$\delta$} (s1) 
		(s2) edge [loop right] node {$\delta$} (s2);
\end{tikzpicture}
}
\hfill
\subfigure[\aut{B}]{
\begin{tikzpicture}
	\node[initial, state] (s0) {};
	\node[state] (s1) [below of=s0] {};

	\path (s0) edge [loop right] node {$\delta$} (s0)
		(s0) edge node [swap] {$a!$} (s1)
		(s1) edge [loop right] node {$\delta$} (s1);
\end{tikzpicture}
}
\hfill
\subfigure[\aut{C}]{
\begin{tikzpicture}
	\node[initial, state] (s0) {};
	\node[state] (s1) [below right  of=s0] {};
	\node[state] (s2) [below left of=s0] {};
	\node[state] (s3) [below of=s1] {};
	\node[state] (s4) [below of=s2] {};
	\node[state] (s5) [below of=s3] {};
	
	\path (s0) edge node {$\delta$} (s1)
		 (s0) edge node [swap] {$a?$} (s2)
		 (s1) edge [loop right] node {$\delta$} (s1) 
		 (s1) edge node {$a?$} (s3)
		 (s2) edge node [swap] {$b!$} (s4)
		 (s3) edge node {$c!$} (s5)
		 (s4) edge [loop left] node {$\delta$} (s4)
		 (s5) edge [loop right] node {$\delta$} (s5);

\end{tikzpicture}
}
\hfill
\subfigure[\aut{D}]{
\begin{tikzpicture}
	\node[initial, state] (s0) {};
	\node[] (s1a) [right of=s0] {};
	\node[state] (s1)  [right of=s1a] {};
	\node[state] (s2)  [right of=s1] {};
	\node[state] (s3)  [below of=s0] {};
	\node[state] (s4)  [below of=s1] {};
	\node[state] (s6)  [below of=s3] {};
	\node[state] (s7)  [below of=s4] {};
	\node[state] (s5)  [left of=s6] {};
	\node[state] (s8)  [right of=s6] {};
	\node[state] (s9)  [below of=s2] {};
	\node[state] (s10)  [below of=s9] {};

	\path 	(s0) edge node {$\delta$} (s1)
				(s0) edge node {$a?$} (s3)
				(s1) edge node {$\delta$} (s2)
				(s1) edge node {$a?$} (s4)
				(s3) edge node [swap] {$b!$} (s5)
				(s3) edge node {$c!$} (s6)
				(s3) edge node {$d!$} (s8)
				(s4) edge node {$b!$} (s7)
				(s4) edge node [swap] {$d!$} (s8)
				(s5) edge [loop below] node {$\delta$} (s5)
				(s6) edge [loop below] node {$\delta$} (s6)
				(s8) edge [loop below] node {$\delta$} (s8)
				(s2) edge node {$a?$} (s9)
				(s9) edge node {$b!$} (s10)
				(s10) edge [loop below] node {$\delta$} (s10)
				(s7) edge [loop below] node {$\delta$} (s7)
				(s2) edge [loop right] node {$\delta$} (s2);
\end{tikzpicture}
}
\caption{The QTSs \aut{A}, \aut{B}, \aut{C} and \aut{D} that do not satisfy rule R1, rule R2, rule R3 and rule R4, respectively.}
\label{fig:qts_rule_violations}
\end{figure}
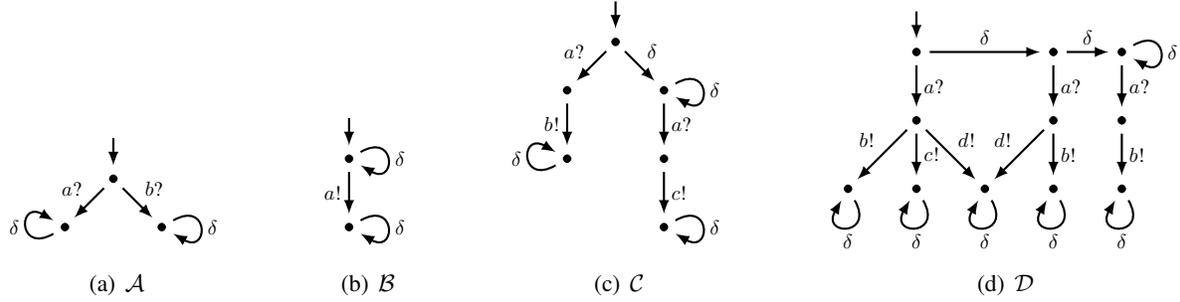

\begin{definition}[Quiescent Transition Systems]
\label{def:qts}
A \emph{Quiescent Transition System} (QTS) is an IOTS \linebreak $\aut{A} = \tuple{S, S^0, \Lin, \Lout \union \set{\delta},\etrans}$, where $\delta \notin \Lin \union \Lout$ is a special output label that is used to denote the observation of quiescence. We define $L = \Lin \union \Lout$, $\Ldeltatau = \Lin \union \Lout \union \set{\delta, \tau}$ and let $\etrans{} \subseteq S \times \Ldeltatau \times S$ be the transition relation. Like regular IOTSs, QTSs must be input-enabled, i.e., $s \trans{a}$ for all $s \in S, a \in \Lin$. Furthermore, we also require the following rules to hold for all states $s, s', s'' \in S$:

\begin{ruledef}{Quiescence should be observable}{if $s$ is quiescent, then $s \trans{\delta}$}
This rule requires that each quiescent state has an outgoing $\delta$-transition. Consider the QTS $\aut{A}$ in Figure~\ref{fig:qts_rule_violations}(a). This QTS does not satisfy this rule, as the topmost state cannot produce any outputs, but neither can execute an outgoing $\delta$-transition.
\end{ruledef}

\begin{ruledef}{No outputs after quiescence}{if $s \trans{\delta} s'$, then $s'$ is quiescent}
This rule ensures that the system is idle after a $\delta$-transition, i.e., it cannot provide an output (except for $\delta$ itself) or execute an internal transition, before another input is provided. In Figure~\ref{fig:qts_rule_violations}(b) the QTS $\aut{B}$ is shown which does not satisfy this rule. From the top-most state it is possible to first observe quiescence (the $\delta$-transition) and after that the $a!$ output, without an intermediate input. Since there is no particular observation duration associated with quiescence, but quiescence rather means that the system idles indefinitely, this is clearly counterintuitive and therefore disallowed.
\end{ruledef}

\begin{ruledef}{Quiescence does not enable new behaviour}{if $s \trans{\delta} s'$, then $\traces{s'} \subseteq \traces{s}$}
Given a state $s'$ of a QTS that is reached from another state $s$ by a $\delta$-transition (i.e., observation of quiescence), this rule demands that any trace that can be executed starting from state $s'$ can also be executed in state $s$, i.e., the observation of quiescence may not introduce any new possible observations. This rule was added to prevent situations like the one depicted in Figure~\ref{fig:qts_rule_violations}(c). For QTS $\aut{C}$ it is possible to observe the output $c!$ (after the input $a?$) after first observing quiescence, but if quiescence is not observed (because, for instance, the input $a?$ was directly given) the output $b!$ will be observed after the input $a?$ instead. Thus, the prior observation of quiescence allows new behaviour to be observed later on, which is counterintuitive. This rule therefore ensures that all behaviour that can be observed after observing quiescence can also be observed before.
\end{ruledef}

\begin{ruledef}{Continued quiescence preserves behaviour}{if $s \trans{\delta} s'$ and $s' \trans{\delta} s''$, then $\traces{s'} = \traces{s''}$}
A QTS $\aut{D}$ that violates this rule is shown in Figure~\ref{fig:qts_rule_violations}(d). From the initial state an observation of quiescence can be made, which then leads to a new state where the trace $ac$ can no longer be observed.  From the latter state another observation of quiescence can be made, which leads to another state where the trace $ad$ can no longer be observed. Rule R3 allows this, but as there is no particular time interval associated with the observation of quiescence, this does not make sense. We therefore have the additional requirement that any observations possible after two (or more) consecutive observations of quiescence should also be possible after a single observation of quiescence, and vice versa.
\end{ruledef}
\end{definition}

Just as for IOTSs, we require QTSs to be convergent. The reason for this is that divergent systems have  states that can execute internal transitions infinitely often and never output anything. Considering such a state quiescent would be nonintuitive, as it is not idle (and might even be able to provide an output action, even though it does not show it). Not considering it quiescent would also be nonintuitive, because of the possibility that no visible behaviour is observed.

Note that the converse of rule R1 is not required, e.g., we do not forbid that a state has both a $\delta$-transition and an output action enabled. This situation can arise during the determinisation of a QTS, as we will see in Section~\ref{sec:deltafication}. However, the $\delta$-transition should still end up in a quiescent state, as required by rule R2. Also note that a trace of a QTS can contain a sequence of $\delta$-actions. Although this might seem odd, it corresponds to
the practical testing scenario of observing a time-out rather than an output more than once in a row.

%Also note that R1 does not require a direct $\delta$-transition, but allows some $\tau$-steps in between. This choice was made to ensure that QTSs are closed under hiding. Take for instance the QTS in Figure~\ref{fig:qts_example}(b). After hiding the $b!$ action, the middle state is quiescent but does not have a direct $\delta$-transition. However, it does still satisfy our more liberal rule R1.

%The reason we require convergence for QTSs is to ensure soundness of rules R1 and R2. For example, consider a QTS with a state that has both an output transition and a $\tau$-loop enabled. As it's possible for the QTS to perform an endless series of $\tau$-transitions in this state, the output might never be observed. Thus, this state can be quiescent in the sense that no output is observed, and hence should have a $\delta$-transition according to rule R1. However, it's also possible for the QTS to execute the output transition in this state, thereby violating rule R2. To prevent such complications, we explicitly disallow divergency.

Since computing trace inclusion is expensive [1], an easier way to ensure that a QTS complies to rule R3 is to make sure the following alternative rule R3$'$ holds for all states $s, s', s'' \in S$.% of the QTS.

\begin{description}
\item[Rule R3$'$:] if $s \trans{\delta} s'$ and $\exists a? \in L_{\text{I}}$ such that $s' \trans{a?} s''$ then also $s \trans{a?} s''$\textbf{.}
\end{description}

\noindent Clearly, any QTS that satisfies rule R3$'$ also satisfies rule R3.

Similarly, conformance to rule R4 for a QTS can be achieved by making sure that the following alternative rule R4$'$ holds for all states $s, s' \in S$ of the QTS.

\begin{description}
\item[Rule R4$'$:] if $s \trans{\delta} s'$ then $s' \trans{\delta} s'$, and if also $s' \trans{\delta} s''$ then $s'' = s'$\textbf{.}
\end{description}

\noindent Clearly, any QTS that satisfies rule R4$'$ also satisfies rule R4.

When comparing the structure of two QTSs \aut{A} and \aut{B}, the notion of isomorphisms can be useful.

\begin{definition}[Isomorphic QTSs]
Two QTSs $\aut{A} = \tuple{\SA, \SstartA, \LinA, \LoutA \union \set{\delta}, \etransA}$ and \linebreak $\aut{B} = \tuple{\SB, \SstartB, \LinB, \LoutB \union \set{\delta}, \etransB}$ are called \emph{isomorphic}, denoted $\aut{A} \isomorphic \aut{B}$, if there exists a bijection $h \colon \SA \rightarrow \SB$ (called an isomorphism) such that the following holds:

\begin{enumerate}
\item for all $s_0 \in \SstartA$ there exists a $t_0 \in \SstartB$ such that $h(s_0) = t_0$, and vice versa;
\item $s \transA{a} s'$ if and only if $h(s) \transB{a} h(s')$, for all $s, s' \in \SA$ and $a \in \LA \union \set{\delta, \tau}$.
\end{enumerate}
\end{definition}

\noindent Thus, two isomorphic QTSs are structurally equivalent.

\subsection{Operations on QTSs}
Since QTSs are a specialisation of IOTSs, all operations that are applicable to IOTSs (such as determinisation, parallel composition and hiding of actions) are also applicable to QTSs. Determinisation for QTSs is exactly the same as for IOTSs, but there are some minor differences for parallel composition and action hiding.

\begin{definition}[Parallel composition of QTSs]
\label{def:parallel_qts}
Let $\aut{A} = \tuple{\SA, S^{0}_{\aut{A}}, \LinA, \LoutA \union \set{\delta}, \etransA}$ and \linebreak $\aut{B} = \tuple{\SB, S^{0}_{\aut{B}}, \LinB, \LoutB \union \set{\delta}, \etransB}$ be two QTSs such that $\LoutA \intersection \LoutB = \emptyset$. The \emph{parallel composition} of \aut{A} and \aut{B} is then the QTS $\AparB= \tuple{\SAparB, S^0_{\AparB}, \LinAparB, \LoutAparB \union \set{\delta}, \etransAparB}$, where $\SAparB = \SA \times \SB$, $S^0_{\AparB} = S^{0}_{\aut{A}} \times S^{0}_{\aut{B}}$, $\LinAparB = (\LinA \union \LinB) \setminus (\LoutA \union \LoutB)$, and $\LoutAparB = \LoutA \union \LoutB$. $\etransAparB$ is defined as follows:
\begin{align*}
\etransAparB \; \; = \; \; & \set{ ((s, t), a?, (s', t')) \where s \transA{a?} s' \land t \transB{a?} t' } \\
\union \; \; & \set{ ((s, t), a!, (s', t')) \where s \transA{a?} s' \land t \transB{a!} t' } \\
\union \; \; & \set{ ((s, t), a!, (s', t')) \where s \transA{a!} s' \land t \transB{a?} t' } \\
\union \; \; & \set{ ((s, t), \delta, (s', t')) \where (s, \delta, s') \in \etransA \land (t, \delta, t') \in \etransB} \\
\union \; \; & \set{ ((s, t), a, (s', t)) \where s \transA{a} s' \land t \in \SB \land a \in \Ltau_{\aut{A}} \setminus \LB } \\
\union \; \; & \set{ ((s, t), a, (s, t')) \where t \transB{a} t' \land s \in \SA \land a \in \Ltau_{\aut{B}} \setminus \LA }
\end{align*}
\end{definition}

Thus, when compared to the parallel composition of regular IOTSs, we have the additional requirement that parallel composed QTSs must synchronise on the $\delta$-action, as the observation of quiescence can be made simultaneously for multiple QTSs. Again, we find that $\LAparB =  \LinAparB \union \LoutAparB = \LA \union \LB$.

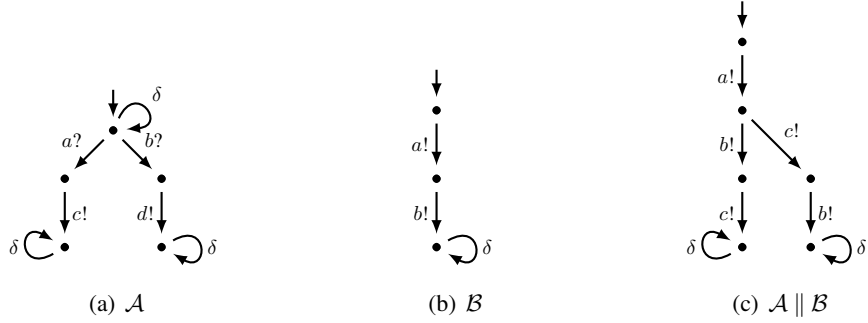
\begin{figure}[t]
\hfill
\subfigure[$\aut{A}$]{
\begin{tikzpicture}
	\node[initial, state] (s0) {};
	\node[state] (s1) [below left of=s0] {};
	\node[state] (s2) [below right of=s0] {};
	\node[state] (s3) [below of=s1] {};
	\node[state] (s4) [below of=s2] {};

	\path (s0) edge [loop above right] node {$\delta$}  (s0) 
		(s0) edge node [swap] {$a?$} (s1)
		(s0) edge node {$b?$} (s2)
		(s1) edge node {$c!$} (s3)
		(s2) edge node [swap] {$d!$} (s4)
		(s3) edge [loop left] node {$\delta$} (s3) 
		(s4) edge [loop right] node {$\delta$} (s4);
\end{tikzpicture}
}
\hfill
\subfigure[$\aut{B}$]{
\begin{tikzpicture}
	\node[initial, state] (s0) {};
	\node[state] (s1) [below of=s0] {};
	\node[state] (s2) [below of=s1] {};

	\path	(s0) edge node [swap] {$a!$} (s1)
			(s1) edge node [swap] {$b!$} (s2)
			(s2) edge [loop right] node {$\delta$} (s2);
\end{tikzpicture}
}
\hfill
\subfigure[$\aut{A} \parallel \aut{B}$]{
\begin{tikzpicture}
	\node[initial, state] (s0) {};
	\node[state] (s1) [below of=s0] {};
	\node[state] (s2) [below of=s1] {};
	\node[state] (s3) [below of=s2] {};
	\node[state] (s4) [right of=s2] {};
	\node[state] (s5) [below of=s4] {};		

	\path	(s0) edge node [swap] {$a!$} (s1)
			(s1) edge node [swap] {$b!$} (s2)
			(s1) edge node {$c!$} (s4)
			(s2) edge node [swap] {$c!$} (s3)
			(s3) edge [loop left] node {$\delta$} (s3)
			(s4) edge node {$b!$} (s5)
			(s5) edge [loop right] node {$\delta$} (s5);
\end{tikzpicture}
}
\hfill {}
\caption{The QTSs $\aut{A}$, $\aut{B}$ and $\aut{A} \parallel \aut{B}$.}
\label{fig:qts_example}
\end{figure}

\begin{example}
See Figure~\ref{fig:qts_example}(a) for the visual representation of a QTS \aut{A} which satifies all the requirements for QTSs listed in Definition \ref{def:qts}. Figure~\ref{fig:qts_example}(b) shows another QTS \aut{B} and Figure~\ref{fig:qts_example}(c) shows the parallel composition of the QTSs \aut{A} and \aut{B}. \qed
\end{example}

\begin{definition}[Action hiding in QTSs]
\label{def:qts_action_hiding}
Let $\aut{A} = \tuple{S, S^0, \Lin, \Lout \union \set{\delta}, \etransA}$ be a QTS and $H \subseteq \Lout$ a set of labels, then one can \emph{hide} $H$ in \aut{A} to obtain the IOTS $\hide{\aut{A}, H} = \tuple{S, S^{0}, \Lin, (\Lout \setminus H) \union \set{\delta}, \etransHide}$, where $\etransHide {} = \set{ (s, a, s') \in \etransA \where a \notin H} \union \set{ (s, \tau, s') \in S \times \set{\tau} \times S \where \exists a \in H \suchthat (s, a, s') \in \etransA }$.
\end{definition}

We do not allow the special output label $\delta$ to be hidden, as this label doesn't represent a specific output but rather (the observation of) a lack of outputs. Furthermore, as for IOTSs, we do not allow action hiding to lead to divergent QTSs, i.e., hiding may not lead to the creation of $\tau$-loops.

\subsection{Properties of QTSs}
\label{sec:qts_properties}
In this section, we present several interesting properties of QTSs. First of all, it turns out 
that our model is closed under all operations defined thus far: determinisation, action hiding
and parallel composition. Therefore, these operations are indeed well-defined for QTSs.

\newcommand{\theoremQTSClosedUnderAllOperations}{%
QTSs are closed under determinisation, action hiding and parallel composition. 
Hence, given two QTSs \aut{A}, \aut{B} and a set of labels $H \subseteq \LoutA$, also 
\deter{\aut{A}}, \hide{\aut{A}, H} and $\AparB$ are QTSs.
}
\begin{theorem}
\label{thm:qts_closed_under_all_operations}
\theoremQTSClosedUnderAllOperations
\end{theorem}

%\newcommand{\theoremQTSClosedUnderDeterminisation}{%
%QTSs are closed under determinisation, i.e., given a QTS \aut{A}, \deter{\aut{A}} is a QTS.
%}
%
%\begin{theorem}
%\label{thm:qts_closed_under_determinisation}
%\theoremQTSClosedUnderDeterminisation
%\end{theorem}
%
%\newcommand{\theoremQTSClosedUnderActionHiding}{%
%QTSs are closed under action hiding, i.e., given a QTS $\aut{A}$ and a set of labels $H \subseteq \LoutA$, \hide{\aut{A}, H} is a QTS.
%}
%
%\begin{theorem}
%\label{thm:qts_closed_under_action_hiding}
%QTSs are closed under action hiding, i.e., given a QTS $\aut{A}$ and a set of labels $H \subseteq \LoutA$, \hide{\aut{A}, H} is a QTS.
%\end{theorem}
%
%\newcommand{\theoremQTSClosedUnderComposition}{%
%QTSs are closed under parallel composition, i.e., given two QTSs \aut{A} and \aut{B}, $\AparB$ is a QTS.
%}
%
%\begin{theorem}
%\label{thm:qts_closed_under_composition}
%\theoremQTSClosedUnderComposition
%\end{theorem}

We also provide two results concerning the traces of parallel compositions of QTSs, generalising the corresponding properties of IOTSs as given in Section~\ref{sec:iots_properties}. First, parallel composition does not introduce new traces when projecting on the alphabet of either one of the components. That is, when disregarding the actions of component \aut{B} in the traces of $\AparB$, the resulting set of traces is a subset of the traces of \aut{A}. It then quite easily follows that, when two parallel QTSs have the same alphabet (and hence synchronise on all actions), we obtain a subset of the intersection of their individual traces.

\newcommand{\propositionQTSParallelTraceInclusion}{%
Given two QTSs \aut{A} and \aut{B}, we have \mbox{$\traces{\AparB} \project (\LA \union \set{\delta}) \subseteq \traces{\aut{A}}$} and \linebreak $\traces{\AparB} \project (\LB \union \set{\delta}) \subseteq \traces{\aut{B}}$.
}

\begin{proposition}
\label{prop:qts_parallel_trace_inclusion}
\propositionQTSParallelTraceInclusion
\end{proposition}

\newcommand{\propositionQTSParallelTraceEquivalence}{%
Given two QTSs \aut{A}, \aut{B} with $\LA = \LB$, we have $\traces{\aut{A} \parallel \aut{B}} = \traces{\aut{A}} \intersection \traces{\aut{B}}$.
}

\begin{proposition}
\label{prop:qts_parallel_trace_equivalence}
\propositionQTSParallelTraceEquivalence
\end{proposition}

\section{From IOTS to QTS: deltafication}\label{sec:deltafication}

Usually, the specification and implementation of a system (under development) are given as IOTSs, rather than QTSs. During testing, however, we typically observe the outputs of the system generated in response to inputs from the environment; thus, it is useful to be able to refer to the absence of outputs (i.e., quiescence) explicitly. Hence, we need a way to convert an IOTS to a QTS that captures all possible observations of it, including quiescence; this conversion is called deltafication and is described in \cite{Tretmans96, Tretmans96b, Tretmans2008}. First, however, we need to introduce an additional condition C1 for IOTSs, for every $s,s' \in S$:

\begin{description}
\item[Condition C1:] if $s \trans{\delta} s'$, then for all $\sigma \in \traces{s'}$:
\[
\exists t' \in \reach{s', \sigma} \suchthat t' \text{ is quiescent } \land t' \ntrans{\delta}{} \Rightarrow \forall t \in \reach{s, \sigma} \suchthat t \text{ is quiescent}\land t \ntrans{\delta} 
\]
\end{description}

Condition C1 requires that if any trace $\sigma \in \traces{s'}$, when executed from $s'$, can lead to a state that is quiescent and cannot execute a $\delta$-transition, then it must always lead to a state that is quiescent and cannot execute a $\delta$-transition when executed from $s$.
This condition is weaker than R1, and allows us to determine the deltafication of systems that already contain some $\delta$-transitions without requiring a $\delta$-transition from every quiescent state. Note that any IOTS without $\delta$-transitions vacuously satisfies C1.
\begin{definition}[Deltafication]
\label{def:deltafication}
Given an IOTS $\aut{A} = \tuple{S, S^0, \Lin, \Lout, \etransA}$ that for all $s, s' \in S$ satisfies deltafication condition C1, and rules R2, R3 and R4 (see Definition~\ref{def:qts}), we define the \emph{deltafication} of \aut{A} as the QTS $\deltaf{\aut{A}} = \tuple{S, S^0, \Lin, \Lout \union \set{ \delta }, \etransDelta}$ where $\etransDelta {} = {} \etransA {} \union \set{ (s, \delta, s) \in S \times \set{ \delta } \times S \where \linebreak s \; \text{is quiescent} \land s \ntransA{\delta} }$.
\end{definition}

\begin{figure}
\hfill
\subfigure[\aut{A}]{
\begin{tikzpicture}
	\node[initial, state] (s0) {};
	\node[state] (s1) [below left of=s0] {};
	\node[state] (s2) [below right of=s0] {};
	\node[state] (s3) [below left of=s2] {};
	\node[state] (s4) [below right of=s2] {};
	\node[state] (s5) [below of=s3] {};
	\node[state] (s6) [below of=s4] {};
	
	\path	(s0) edge node [swap] {$b!$} (s1)
			(s0) edge node {$a?$} (s2)
			(s2) edge node {$\tau$} (s4)
			(s4) edge node {$\tau$} (s3)
			(s3) edge node [swap] {$c?$} (s5)
			(s4) edge node {$d?$} (s6);
\end{tikzpicture}
}
\hfill
\subfigure[\deltaf{\aut{A}}]{
\begin{tikzpicture}
	\node[initial, state] (s0) {};
	\node[state] (s1) [below left of=s0] {};
	\node[state] (s2) [below right of=s0] {};
	\node[state] (s3) [below left of=s2] {};
	\node[state] (s4) [below right of=s2] {};
	\node[state] (s5) [below of=s3] {};
	\node[state] (s6) [below of=s4] {};
	
	\path	(s0) edge node [swap] {$b!$} (s1)
			(s1) edge [loop left] node {$\delta$} (s1)
			(s0) edge node {$a?$} (s2)
			(s2) edge node {$\tau$} (s4)
			(s4) edge node {$\tau$} (s3)
			(s3) edge [loop left] node {$\delta$} (s3)
			(s3) edge node {$c?$} (s5)
			(s5) edge [loop left] node {$\delta$} (s5)
			(s4) edge node [swap] {$d?$} (s6)
			(s6) edge [loop right] node {$\delta$} (s6);
\end{tikzpicture}
}
\hfill
\subfigure[\aut{B}]{
\begin{tikzpicture}[node distance=1.5cm]
	\tikzstyle{every state}=[draw=black,line width=1pt,fill=white,minimum size=5pt]

	\node[initial,state] (s0) {\footnotesize $s_0$};
	\node[state] (s1) [right of=s0] {\footnotesize $s_1$};
	\node[state] (s3) [below of=s0] {\footnotesize $s_3$};
	\node[state] (s2) [left  of=s3] {\footnotesize $s_2$};
	\node[state] (s4) [below of=s1] {\footnotesize $s_4$};
	\node[state] (s6) [below of=s3] {\footnotesize $s_6$};
	\node[state] (s5) [left  of=s6] {\footnotesize $s_5$};
	\node[state] (s7) [below of=s4] {\footnotesize $s_7$};
	\node[state] (s8) [below of=s6] {\footnotesize $s_8$};
	\node[state] (s9) [below of=s7] {\footnotesize $s_9$};

	\path (s0) edge node {$\delta$}  (s1)
			 (s0) edge node {$a?$}  (s3)
			 (s1) edge node {$a?$}  (s4)
			 (s0) edge [swap] node {$a?$}  (s2)
			 (s3) edge node {$b?$} (s6)
			 (s3) edge [swap] node {$c!$} (s5)
			 (s6) edge node {$c!$} (s8)
			 (s4) edge node {$b?$}  (s7)
			 (s7) edge node {$c!$}  (s9);
\end{tikzpicture}
}
\hfill
\subfigure[\deltaf{\aut{B}}]{
\begin{tikzpicture}[node distance=1.5cm]
	\tikzstyle{every state}=[draw=black,line width=1pt,fill=white,minimum size=5pt]

	\node[initial,state] (s0) {\footnotesize $s_0$};
	\node[state] (s1) [right of=s0] {\footnotesize $s_1$};
	\node[state] (s3) [below of=s0] {\footnotesize $s_3$};
	\node[state] (s2) [left  of=s3] {\footnotesize $s_2$};
	\node[state] (s4) [below of=s1] {\footnotesize $s_4$};
	\node[state] (s6) [below of=s3] {\footnotesize $s_6$};
	\node[state] (s5) [left  of=s6] {\footnotesize $s_5$};
	\node[state] (s7) [below of=s4] {\footnotesize $s_7$};
	\node[state] (s8) [below of=s6] {\footnotesize $s_8$};
	\node[state] (s9) [below of=s7] {\footnotesize $s_9$};

	\path (s0) edge node {$\delta$}  (s1)
			 (s1) edge [loop right] node {$\delta$} (s1)
			 (s4) edge [loop right] node {$\delta$} (s4)
			 (s9) edge [loop right] node {$\delta$} (s9)
			 (s8) edge [loop left ] node {$\delta$} (s8)
			 (s5) edge [loop below] node {$\delta$} (s5)
			 (s2) edge [loop above] node {$\delta$} (s2)
			 (s0) edge node {$a?$}  (s3)
			 (s1) edge node {$a?$}  (s4)
			 (s0) edge [swap, pos=0.3] node {$a?$}  (s2)
			 (s3) edge node {$b?$} (s6)
			 (s3) edge [swap, pos=0.3] node {$c!$} (s5)
			 (s6) edge node {$c!$} (s8)
			 (s4) edge node {$b?$}  (s7)
			 (s7) edge node {$c!$}  (s9);
\end{tikzpicture}
}
\caption{Deltafications of the IOTSs \aut{A} and \aut{B}.}
\label{fig:deltafication_example}
\end{figure}
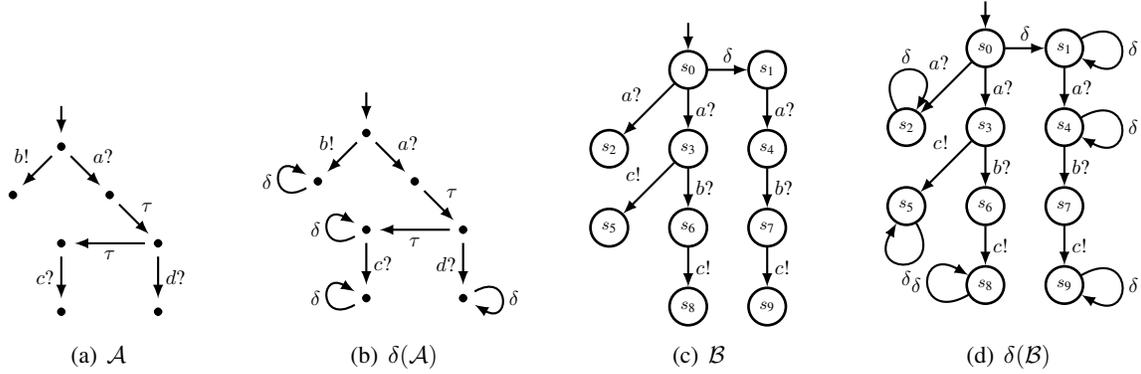

\begin{example}
An IOTS \aut{A} and its deltafication \deltaf{\aut{A}} are shown in Figure~\ref{fig:deltafication_example}(a) and~\ref{fig:deltafication_example}(b), respectively. \qed
\end{example}

\begin{remark}
To see why condition C1 is necessary, consider the IOTS \aut{B} and its deltafication \deltaf{\aut{B}} shown in Figure~\ref{fig:deltafication_example}(c) and Figure~\ref{fig:deltafication_example}(d), respectively; the states have been labelled for convenience. \aut{B} does not satisfy condition C1, since $s_0 \trans{\delta} s_1$, $s_4 \in \reach{s_1, a}$ and $s_4$ is quiescent and $s_4 \ntrans{\delta}$, but $s_3 \in \reach{s_0, a}$ and $s_3$ is not quiescent. As a consequence, the deltafication \deltaf{\aut{B}} is not a valid QTS: for \deltaf{\aut{B}} we have $a\delta bc \in \traces{s_1}$, but $a\delta bc \notin \traces{s_0}$, thereby violating rule R3.

A more liberal version of C1, where the second quantification is changed to an existential one, would not be strong enough to prevent this: it would not forbid this example, as $s_2 \in \reach{s_0, a}$ is quiescent and cannot do a $\delta$-transition.
\end{remark}

\subsection{Validity of deltafication}
Now, we present several interesting properties regarding the deltafication of IOTSs and QTSs. First,
we show that deltafication indeed yields a valid QTS, and that it is idempotent.

\newcommand{\lemmaDeltaficationYieldsValidQTS}{%
Given an IOTS \aut{A} that satisfies condition C1 and rules R2, R3 and R4, \deltaf{\aut{A}} is a QTS.
}

\begin{lemma}
\label{lemma:deltafication_yields_valid_qts}
\lemmaDeltaficationYieldsValidQTS
\end{lemma}

\newcommand{\propositionDeltaficationIsIdempotent}{%
Deltafication is idempotent, i.e., given an IOTS \aut{A} that satisfies condition C1 and rules R2, R3 and R4, we have \deltaf{\deltaf{\aut{A}}} = \deltaf{\aut{A}}.
}

\begin{proposition}
\label{prop:deltafication_is_idempotent}
\propositionDeltaficationIsIdempotent
\end{proposition}

Any IOTS \aut{A} with $\delta \notin \LA$ vacuously satisfies condition C1 and rules R2, R3 and R4. Therefore, the following theorem follows directly from Lemma~\ref{lemma:deltafication_yields_valid_qts}.

\begin{theorem}
Given an IOTS \aut{A} such that $\delta \notin \LA$, \deltaf{\aut{A}} is a QTS.
\end{theorem}

By Definition~\ref{def:qts}, QTSs are IOTSs that satisfy rules R1, R2, R3 and R4. Since every state $s$ 
in a QTS enables at least one output action or $\delta$ (due to rule R1), it never occurs that $s$ is quiescent and does not enable a $\delta$-transition, and hence every QTS satisfies condition C1 vacuously. 

By Lemma~\ref{lemma:deltafication_yields_valid_qts}, this immediately implies the following theorem.

\begin{theorem}
QTSs are closed under deltafication, i.e., given a QTS \aut{A}, \deltaf{\aut{A}} is also a QTS.
\end{theorem}

\subsection{Commutativity results}\label{sec:comm}
In this section we investigate the commutativity of deltafication with determinisation, action hiding and parallel composition. We will show that parallel composition can safely be swapped with deltafication, but that determinisation has to precede deltafication to get sensible results. Also, we show that action hiding does not commute with deltafication.

\begin{proposition}
Deltafication and determinisation do not commute, i.e., given an IOTS \aut{A} that satisfies condition C1 and rules R2, R3 and R4, it is not necessarily the case that $\deter{\deltaf{\aut{A}}} \isomorphic \deltaf{\deter{\aut{A}}}$.
\end{proposition}

\begin{proof}
Observe the IOTS \aut{A}, its determinisation \deter{\aut{A}} and deltafication \deltaf{\aut{A}} in Figure~\ref{fig:determinisation_and_deltafication}(a,b,c). Clearly, the deltafication of the determinisation of \aut{A} (i.e., \deltaf{\deter{\aut{A}}}), shown in Figure~\ref{fig:determinisation_and_deltafication}(d), results in an incorrect observation automaton, as it does not model the fact that in the nondeterministic QTS \deltaf{\aut{A}} quiescence may be observed after an initial $a?$ input, as required by rule R1.

\begin{figure}
\subfigure[\aut{A}]{
\begin{tikzpicture}
	\node[initial, state] (s0) {};
	\node[state] (s1) [below left of=s0] {};
	\node[state] (s2) [below right of=s0] {};
	\node[state] (s3) [below of=s1] {};
	\node[state] (s4) [below of=s2] {};

	\path (s0) edge node [swap]  {$a?$}  (s1)
		(s0) edge node {$a?$}  (s2)
		(s1) edge node [swap] {$b!$}  (s3)
		(s2) edge node {$a?$} (s4);
\end{tikzpicture}
}
\hfill
\subfigure[$\deter{\aut{A}}$]{
\begin{tikzpicture}
	\node[initial, state] (s0) {};
	\node[state] (s1) [below of=s0] {};
	\node[state] (s2) [below left of=s1] {};
	\node[state] (s3) [below right of=s1] {};

	\path (s0) edge node {$a?$}  (s1)
		(s1) edge node [swap] {$b!$}  (s2)
		(s1) edge node {$a?$} (s3);
\end{tikzpicture}
}
\hfill
\subfigure[$\deltaf{\aut{A}}$]{
\begin{tikzpicture}
	\node[initial, state] (s0) {};
	\node[state] (s1) [below left of=s0] {};
	\node[state] (s2) [below right of=s0] {};
	\node[state] (s3) [below of=s1] {};
	\node[state] (s4) [below of=s2] {};

	\path	(s0) edge node [swap]  {$a?$}  (s1)
			(s0) edge [loop above right] node {$\delta$} (s0)
			(s0) edge node {$a?$}  (s2)
			(s1) edge node {$b!$}  (s3)
			(s2) edge node [swap] {$a?$} (s4)
			(s2) edge [loop right] node {$\delta$} (s2)
			(s3) edge [loop below] node {$\delta$} (s3)
			(s4) edge [loop below] node {$\delta$} (s4);
\end{tikzpicture}
}
\hfill
\subfigure[$\deltaf{\deter{\aut{A}}}$]{
\begin{tikzpicture}
	\node[initial, state] (s0) {};
	\node[state] (s1) [below of=s0] {};
	\node[state] (s2) [below left of=s1] {};
	\node[state] (s3) [below right of=s1] {};

	\path (s0) edge [loop right] node {$\delta$} (s0) 
		(s0) edge node [swap]  {$a?$}  (s1)
		(s1) edge node [swap] {$b!$}  (s2)
		(s2) edge [loop below] node {$\delta$} (s2)
		(s1) edge node {$a?$} (s3)
		(s3) edge [loop below] node {$\delta$} (s3);
\end{tikzpicture}
}
\hfill
\subfigure[$\deter{\deltaf{\aut{A}}}$]{
\begin{tikzpicture}
	\node[initial, state] (s0) {};
	\node[state] (s1) [below of=s0] {};
	\node[state] (s2) [below of=s1] {};
	\node[state] (s3) [left of=s2] {};
	\node[state] (s4) [right of=s2] {};
	
	\path (s0) edge [loop right] node {$\delta$} (s0)
		(s0) edge node [swap] {$a?$} (s1)
		(s1) edge node [swap] {$b!$} (s3)
		(s3) edge [loop below] node {$\delta$} (s3)
		(s1) edge node [swap] {$\delta$} (s2)
		(s1) edge node {$a?$} (s4)
		(s2) edge node [swap] {$a?$} (s4)
		(s2) edge [loop below] node {$\delta$} (s2)
		(s4) edge [loop below] node {$\delta$} (s4);
\end{tikzpicture}
}
\caption{%The IOTS \aut{A} along with its determinisation \deter{\aut{A}} and deltafication \deltaf{\aut{A}}.
The determinisation and deltafication of IOTS \aut{A} do not commute.}
\label{fig:determinisation_and_deltafication}
\end{figure}
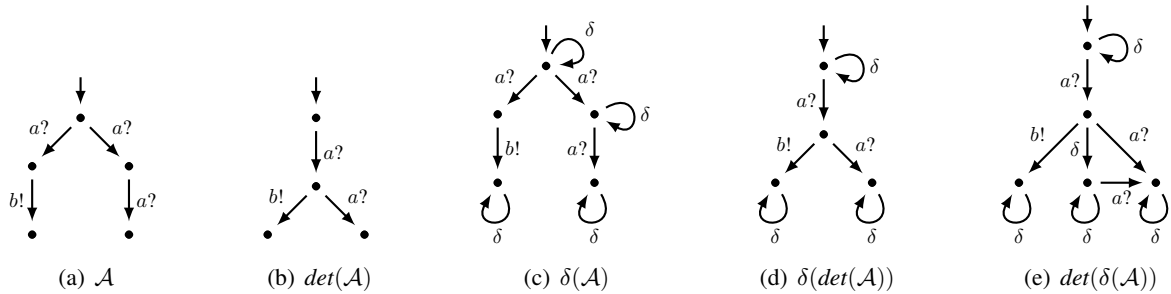

Contrary to the deltafication of the determinisation of \aut{A}%(i.e., \deltaf{\deter{\aut{A}}}), shown in Figure~\ref{fig:determinisation_and_deltafication}(d)
, the determinisation of the deltafication of \aut{A} (i.e., \deter{\deltaf{\aut{A}}}), which is shown in Figure~\ref{fig:determinisation_and_deltafication}(e), does preserve the fact that quiescence may be observed after an initial $a?$ input. This shouldn't come as a surprise, since for any IOTS \aut{A} the determinisation \deter{\aut{A}} is trace equivalent to the original automaton, as was observed earlier. \end{proof}

Thus, when transforming a nondeterministic IOTS \aut{A} to a deterministic QTS, one should take care to first derive \deltaf{\aut{A}} and afterwards determinise to obtain \deter{\deltaf{\aut{A}}}.

The following results show that deltafication does commute with both action hiding and parallel composition.
For action hiding this is trivial. After all, hiding only renames output actions to $\tau$
and deltafication only adds $\delta$-loops to states that have 
no outgoing output transitions, no outgoing $\tau$-transitions and no outgoing $\delta$-transition. Hence, 
they work on disjoint sets of states; commutativity is therefore immediate.  

\begin{theorem}
Deltafication and action hiding commute, i.e., given an IOTS \aut{A} that satisfies condition C1 and rules R2, R3 and R4, and a set of labels $H \subseteq \LoutA$, we have $\deltaf{\hide{\aut{A}, H}} \isomorphic \hide{\deltaf{\aut{A}}, H}$.
\end{theorem}

\newcommand{\theoremDeltaficationAndParallelCompositionCommute}{%
Deltafication and parallel composition commute, i.e., given two IOTSs \aut{A} and \aut{B} with $\LoutA \intersection \LoutB = \emptyset$ that satisfy condition C1 and rules R2, R3 and R4, we have $\deltaf{\AparB} \isomorphic \deltaf{\aut{A}} \parallel \deltaf{\aut{B}}$.
}

\begin{theorem}
\label{thm:deltafication_and_parallel_composition_commute}
\theoremDeltaficationAndParallelCompositionCommute
\end{theorem}

These results are vital, as they allow great modelling flexibility. After all,
hiding and parallel composition are often already applied to the IOTSs that describe a specification 
and its implementation. We now showed that this yields the same QTSs as in case these operations are applied after deltafication.

\section{Application to testing}\label{sec:testing}

\newcommand{\ioco}{\texttt{ioco}}
\newcommand{\iocorel}{\mathrel{\sqsubseteq_\ioco}}
\newcommand{\out}{\mathit{out}}
\newcommand{\aimpl}{\aut{A}_{\text{impl}}}
\newcommand{\aspec}{\aut{A}_{\text{spec}}}

Our main motivation for introducing and studying the QTS model was to enable a clean theoretical framework for model-based testing. In this section, we illustrate how the model can be incorporated in the \texttt{ioco} (input-output conformance) testing theory~\cite{Tretmans2008}. 

\subsection{A conformance relation based on QTSs}
To interpret the results of testing, we need to know which implementations are considered correct.
For this, we use a conformance relation, such as \ioco, that relates specifications to implementations if and only if the latter is `correct' with respect to the former. For \ioco, this is the case if the implementation never provides an unexpected output when it is only fed inputs that are allowed according to the specification. In this setting, an unexpected absence of outputs of the implementation is also considered to be unexpected output. This can be formalised very nicely using QTSs, as they already model the expected absence of outputs by explicit $\delta$-transitions.

\begin{definition}{}

Let $\aimpl, \aspec$ be QTSs over the same alphabet $\Lout \union \Lin \union \{\delta\}$.
%,and let $\aut{A}$ be input-enabled. 
Then
\[
   \aimpl \iocorel \aspec \text{ if and only if }
       \forall \sigma \in \traces{\aspec} \suchthat \out_{\aimpl}(\sigma) \subseteq \out_{\aspec}(\sigma),
\]
where $\out_\aut{A}(\sigma) = \{a! \in \Lout \union \{\delta\} \mid \sigma a! \in \traces{\aut{A}}\}$.
\end{definition}
Since we require all QTSs to be input-enabled, it is easy to see that \ioco-conformance precisely corresponds to traditional trace inclusion over QTSs.

\begin{example}
Consider the specification $\aspec$ given in Figure~\ref{fig:ioco_example}. It allows the initial state to either be quiescent, output an $a!$ or output a $b!$. We present four implementations.
The first two implementations are \ioco-correct with respect to $\aspec$: although they omit some of the traces of the specification, they never provide an unexpected output after a trace that is in the specification. The third implementation is erroneous since it can provide a $d!$ output from the initial state, while the specification does not allow this. The fourth implementation is erroneous since it is unexpectedly quiescent after the trace $c?$.
\qed

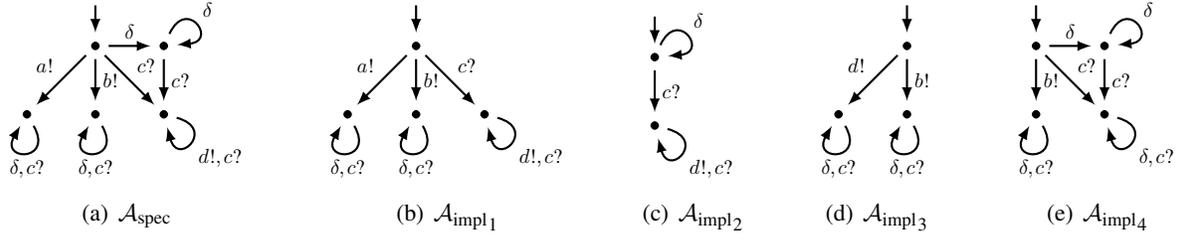
\begin{figure}
\subfigure[$\aspec$]{
\begin{tikzpicture}
	\node[initial, state] (s0) {};
	\node[state] (s1) [right of=s0] {};
	\node[state] (s2) [below of=s1] {};
	\node[state] (s3) [below of=s0] {};
	\node[state] (s4) [left  of=s3] {};

	\path		(s0) edge node {$\delta$}  (s1)
				(s1) edge node {$c?$}  (s2)
				(s0) edge node {$c?$}  (s2)
				(s0) edge node {$b!$}  (s3)
				(s0) edge node [swap] {$a!$}  (s4)
				(s1) edge [loop above right] node {$\delta$}  (s1)
				(s2) edge [loop below right] node {$d!, c?$}  (s2)
				(s3) edge [loop below] node {$\delta, c?$}  (s3)
				(s4) edge [loop below] node {$\delta, c?$}  (s4);
\end{tikzpicture}
}
\hfill
%\subfigure[${\aimpl}_1$]{
%\input{figures/ioco_impl_1.tikz}
%}
%\hfill
\subfigure[${\aimpl}_1$]{
\begin{tikzpicture}
	\node[initial, state] (s0) {};
	\node[] (s1) [right of=s0] {};
	\node[state] (s2) [below of=s1] {};
	\node[state] (s3) [below of=s0] {};
	\node[state] (s4) [left  of=s3] {};

	\path		(s0) edge node {$c?$}  (s2)
				(s0) edge node {$b!$}  (s3)
				(s0) edge node [swap] {$a!$}  (s4)
				(s2) edge [loop below right] node {$d!, c?$}  (s2)
				(s3) edge [loop below] node {$\delta, c?$}  (s3)
				(s4) edge [loop below] node {$\delta, c?$}  (s4);
\end{tikzpicture}
}
\hfill
\subfigure[${\aimpl}_2$]{
\begin{tikzpicture}
	\node[initial, state] (s0) {};
	\node[state] (s1) [below of=s0] {};

	\path		(s0) edge [loop above right] node {$\delta$}  (s0)
				(s0) edge node {$c?$}  (s1)
				(s1) edge [loop below right] node {$d!, c?$}  (s1);
\end{tikzpicture}
}
%\\ 
%\phantom{\ }
\hfill
\subfigure[${\aimpl}_3$]{
\begin{tikzpicture}
	\node[initial, state] (s0) {};
	\node[state] (s3) [below of=s0] {};
	\node[state] (s4) [left  of=s3] {};

	\path		(s0) edge node {$b!$}  (s3)
				(s0) edge node [swap] {$d!$}  (s4)
				(s3) edge [loop below] node {$\delta, c?$}  (s3)
				(s4) edge [loop below] node {$\delta, c?$}  (s4);
\end{tikzpicture}
}
\hfill
\subfigure[${\aimpl}_4$]{
\begin{tikzpicture}
	\node[initial, state] (s0) {};
	\node[state] (s1) [right of=s0] {};
	\node[state] (s2) [below of=s1] {};
	\node[state] (s3) [below of=s0] {};

	\path		(s0) edge node {$\delta$}  (s1)
				(s1) edge node {$c?$}  (s2)
				(s0) edge node {$c?$}  (s2)
				(s0) edge node {$b!$}  (s3)
				(s1) edge [loop above right] node {$\delta$}  (s1)
				(s2) edge [loop below right] node {$\delta, c?$}  (s2)
				(s3) edge [loop below] node {$\delta, c?$}  (s3);
\end{tikzpicture}
}
\caption{A specification with two correct and two erroneous implementations.}
\label{fig:ioco_example}
\end{figure}

\end{example}
Note that QTSs allowed us in this example to explicitly model the fact that both quiescence and some output actions are considered correct behaviour of a system. Also, note that the unexpected quiescence of the fourth implementation is clearly marked by a $\delta$-transition in the QTS. 

\subsection{Testing using QTSs}

Using the notion of \ioco-correspondence, it is quite easy to derive test cases for QTSs. Basically, at each point in time we choose to either try to provide an input, observe the behaviour of the system or stop testing. As long as the trace we obtain in this way (including the $\delta$-actions) is also a trace of the specification, the implementation is correct. Due to the explicit presence of quiescence in the QTS model of the specification, it is easy to see that this straightforward way of testing precisely corresponds to checking \ioco-conformance.

\section{Conclusions and Future Work}\label{sec:conclusions}

We introduced the notion of quiescent transition systems (QTSs), explicitly modelling the 
absence of outputs as a first-class citizen. We provided four restrictions for QTSs, 
to eliminate counterintuitive behaviours. Also, we defined the common automaton operations
--- parallel composition, determinisation and action hiding --- directly on QTSs, and showed 
that all of our restrictions are indeed preserved by the operations. We presented a way to
obtain a QTS from a traditional input-output transition system (IOTS), even allowing the situation
in which the IOTS already partially models quiescence. Finally, we illustrated how our novel 
theory of QTSs can be used to greatly simplify the theory of model-based testing, defining the 
conformance relation \texttt{ioco} in terms of QTSs.

So far, we only allowed input-enabled and convergent QTSs; i.e., systems that cannot perform an endless series of unobservable transitions. Future work will focus on 
extending our framework to divergent systems that are not necessarily input-enabled. Also, we plan on linking QTSs to 
timed automata, to explicitly represent $\delta$-transitions as finite timeouts,
bridging the gap between formal and practical testing.

\paragraph{Acknowledgements}
This research has been partially funded by NWO under grants 612.063.817 \linebreak (SYRUP) and Dn 63-257 (ROCKS).

\bibliographystyle{eptcs}
\bibliography{report}

\clearpage

\end{document}